\documentclass[preprint,12pt,compress]{elsarticle}
\usepackage{lineno}
\usepackage{amsmath,amsfonts}
\usepackage[ruled,linesnumbered,lined]{algorithm2e}
\usepackage{array}
\usepackage{textcomp}
\usepackage{stfloats}
\usepackage{url}
\usepackage{verbatim}
\usepackage{graphicx}
\usepackage{makecell}
\usepackage{epstopdf}
\usepackage{dcolumn}
\usepackage{bm}
\usepackage{natbib}
\setcitestyle{numbers,square}
\usepackage{braket}
\usepackage{color}
\usepackage{balance}
\usepackage{booktabs}
\usepackage{amssymb}
\usepackage{amsthm}
\usepackage{caption}
\usepackage{tikz}
\usepackage{adjustbox}
\usepackage{titlesec}
\usepackage{quantikz}
\usepackage{hyperref}

\theoremstyle{plain}
\newtheorem{theorem}{Theorem}

\newtheorem{lemma}{Lemma}
\newtheorem{proposition}{Proposition}

\theoremstyle{definition}
\newtheorem{definition}{Definition}

\newtheorem{remark}{Remark}

\begin{document}

\begin{frontmatter}



\title{Distributed quantum-classical hybrid algorithm for solving $K$-SAT problem} 


\author[a,b]{Huaijing Huang}
\author[a,b,d]{Daowen Qiu\corref{mycorrespondingauthor}}
\ead{issqdw@mail.sysu.edu.cn}
\cortext[mycorrespondingauthor]{Corresponding author}
\author[c,d]{Le Luo}
\author[e]{Paulo Mateus}
\affiliation[a]{organization={School of Computer Science and Engineering},
            addressline={Sun Yat-sen University}, 
            city={Guangzhou},
            postcode={510006}, 
            country={China}}
\affiliation[b]{organization={The Guangdong Key Laboratory of Information Security Technology},
            addressline={Sun Yat-sen University}, 
            city={Guangzhou},
            postcode={510006}, 
            country={China}}
\affiliation[c]{organization={School of Physics and Astronomy},
            addressline={Sun Yat-sen University}, 
            city={Zhuhai},
            postcode={519082}, 
            country={China}}
\affiliation[d]{organization={Shenzhen Research Institute of Sun Yat-Sen University},
            addressline={Sun Yat-sen University}, 
            city={Shenzhen},
            postcode={518057}, 
            country={China}}
 \affiliation[e]{organization={Departamento de Matemática, Instituto Superior Técnico},
            addressline={Instituto de Telecomunicações}, 
            city={Lisbon},
            postcode={1049-001}, 
            country={Portugal}}
\begin{abstract}
Recently, Dunjko et al.(PRL, 2018) proposed an algorithm for accelerating the solution of 3-satisfiability problems using a small-scale quantum computer. In this paper, we design a distributed quantum-classical hybrid algorithm for solving $K$-satisfiability problems.    Under resource-constrained conditions, our algorithm achieves a significant acceleration in the core term of the exponential time complexity.  The proposed algorithm is a generalization of the algorithm by Dunjko et al. Compared with their algorithm, our algorithm requires a smaller number of qubits. More importantly, the proposed algorithm does not rely on any quantum communication.
\end{abstract}


\end{frontmatter}



\section{Introduction}
The speedup potential of quantum computing stems from its superposition. However, quantum computers in the NISQ era\cite{preskill2018quantum} often struggle to realize this speedup due to a limited number of qubits and circuits with excessive depth. Distributed quantum computing is a promising approach to address these challenges. Currently, a series of important distributed quantum algorithms have been proposed, including distributed Grover's search algorithms\cite{Qiu2022}, distributed Simon's algorithm\cite{TAN_quantum_2022}, and distributed Shor's algorithm\cite{Xiao2023}, etc. Among these algorithms, some require quantum communication, while others do not. Additionally, regarding the errors arising from distributed computing, Qiu et al.  proposed an error correction scheme that can effectively reduce errors in Ref.\cite{qiu2025universal}.

Boolean Satisfiability (SAT) is an important problem in classical computation and mathematics. Many practical problems can be reduced to SAT problems for solving, such as scheduling problems, planning problems, graph coloring, and other combinatorial optimization problems. The $K$-SAT problem refers to SAT instances in which each clause has at most $K$ literals. The 3-SAT problem is the most representative class of SAT problems, and every SAT problem can be reduced to a 3-SAT problem in polynomial time. Since the 3-SAT problem is NP-complete\cite{cook2023complexity}, many solving algorithms require exponential time in the worst case.

Utilizing quantum computing  to solve SAT problems is a highly active area of research. Currently, many quantum-classical hybrid algorithms for solving SAT problems have been proposed. Ambainis\cite{ambainis2004quantum}, based on the Schöning’s algorithm\cite{schoning1999probabilistic}, proposed an algorithm that uses amplitude amplification to speed up solving SAT problems. Zhang et al.\cite{zhang2020procedure} employ the DPLL algorithm and Grover’s algorithm to solve 3-SAT problems. Eshaghian et al.\cite{eshaghian2025runtime}, taking into account the relationship between coherence time and algorithm runtime, proposed a hybrid algorithm that combines the Schöning’s algorithm and Grover’s algorithm to solve $K$-SAT problems. Varmantchaonala et al. \cite{varmantchaonala2023quantum}, by adjusting parameters, proposed a Grover-based method to reduce the search space of SAT problems. Dunjko et al. \cite{dunjko2018computational} proposed a conceptual quantum-classical hybrid algorithm for solving 3-SAT problems under limited qubit conditions.

Currently, most quantum-classical hybrid algorithms for solving SAT problems use Grover's algorithm\cite{Grover1996}, with a small number employing amplitude amplification\cite{brassard2000quantum}. Since in SAT problems we do not know whether a satisfying assignment exists, nor the number of satisfying assignments, applying Grover's algorithm or amplitude amplification can lead to the soufflé problem\cite{brassard1997searching}. The fixed-point quantum search algorithm\cite{yoder2014fixed} effectively avoids the “soufflé” problem. Current quantum computers face inherent limitations in qubit count, coherence time, and connectivity, making it difficult to implement large-scale quantum algorithms on a single processor. Distributed architectures offer a natural pathway to overcome these constraints by interconnecting multiple small-scale quantum nodes, enabling scalable quantum computation that goes beyond the capacity of individual devices. Integrating problem structures with quantum algorithms—rather than adopting a simple hybrid approach—provides a viable means to tackle larger-scale problems using small-scale quantum devices. How to achieve speedup under such resource-constrained conditions is thus a key question on the path toward practical quantum computing. Based on this, we design a distributed quantum-classical hybrid algorithm to solve the $K$-SAT problem.

We first formalize the framework in \cite{dunjko2018computational} and distill it into a structured algorithm. We then generalize this algorithm. In circuit implementation, if a variable appears frequently in the Boolean formula, its computation generally requires the introduction of additional control gates. To reduce control-gate depth and mitigate the "soufflé" problem that arises during amplification, we first give a measure decomposition  the Boolean formula,  and then incorporate a heuristic strategy for designing a distributed quantum–classical hybrid algorithm to solve the $K$-SAT problem. The implementation of this strategy may significantly reduce computational resource consumption.

The paper is structured in the following manner. In Section \ref{a}, quantum-classical hybrid algorithms for solving 3-SAT problem and the fixed-point quantum search algorithm will be reviewed briefly. In Section \ref{3}, considering circuit depth and qubit count constraints, we design a distributed quantum-classical hybrid algorithm to solve $K$-SAT problem. We prove the correctness of the algorithm and compare it with existing quantum algorithms.   Conclusions and  outlook are presented in Section \ref{4}.

 \section{Preliminaries}\label{a}
In this section, we first review the  quantum-classical hybrid algorithm proposed by Dunjko et al.\cite{dunjko2018computational} for solving the $3$-SAT problem,  and then present the fixed-point quantum search algorithm of Yoder et al.\cite{yoder2014fixed}.

\subsection{ Quantum-classical hybrid algorithm for 3-SAT problem}
Dunjko et al. utilized small-scale quantum devices to accelerate the solution of the 3-SAT problem. Before introducing their algorithm, we first provide a formal definition of the SAT problem.
\begin {definition}\label{d1}
Let $F=\land_{j =1}^m C_j$ be a Boolean formula consisting of $m$ clauses $C_j=\lor_{i =1}^n l_i^{(j)}$ and $n$ literals, where each literal $l_i$ is either a Boolean variable $x_i$ or its negation $\neg{x_i}$. Each clause must contain at most $K$ literals,  referred to as $K$-SAT.  The $K$-SAT problem is to determine whether there exists an assignment  to the $n$ variables that satisfies $F$.
\end{definition}
$F$ can be viewed as a Boolean function on $\{0,1\}^n\rightarrow\{0,1\}$; find an assignment $\mathbf{x}=(x_1,x_2,\cdots,x_n)\in\{0,1\}^n$ such that $F(\mathbf{x}) = 1$. Next, we will reduce the $K$-SAT problem to the PBS (Promise-Ball-SAT) problem. First, we introduce the concept of the Hamming ball.
\begin {definition}
Given a bit string $\mathbf{x}\in\{0,1\}^n$, a Hamming ball centered at $\mathbf{x}$ with radius $r\in \mathbb{N}$ is defined as $B_r(\mathbf{x}) = \{y  \in\{0,1\}^n| d_H(x,y) \leq r\}$, where $d_H$ is the Hamming distance. 
\end{definition}
Covering codes\cite{dantsin2002deterministic} are used in the space partitioning algorithm to divide the entire assignment space into multiple Hamming balls of radius $r$, each represented by a center code word.   Construct a covering code $C = \{\mathbf{x}_1, \mathbf{x}_2, \cdots, \mathbf{x}_t\}\subset\{0,1\}^n$ using coding theory, such that:
$$
\forall x \in \{0,1\}^n, \exists \mathbf{x}_i \in C, d_H(x, \mathbf{x}_i) \le r.
$$
A formal definition is given below.
\begin {definition}
If
$$
  \bigcup_{\mathbf{x} \in {C}} B_r(\mathbf{x})= \{0,1\}^n, r\in \mathbb{N}
$$a set ${C} \subseteq \{0,1\}^n$ is called a code of
covering radius $r$. 
\end{definition}
With Hamming balls, all possible assignments can be covered by a collection of Hamming balls, so the problem is reduced to finding a satisfying assignment within the balls. Next, provide the definition of the PBS problem. 
\begin {definition}
Given a $K$-SAT formula $F$ with $n$ variables, a center assignment $\mathbf{x}\in\{0,1\}^n$, and a radius $r\in \mathbb{N}$. It is promised that there exists at least one satisfying assignment within $B_r(\mathbf{x}) $. The PBS problem is to find an assignment that satisfies $F$.
\end{definition}

 The commitment here does not refer to the problem itself assuming a solution must exist, but rather expresses a form of conditional correctness: as long as there exists a satisfying assignment within the search space, the algorithm can find it. Ref. \cite{dunjko2018computational} primarily considers the 3-SAT problem, that is, the case ($K = 3$) in Definition \ref{d1}.
If the commitment of PBS is violated, the algorithm will output false. First, determine the depth $r$ of the ternary tree based on the radius of the PBS problem, non-recursively implement the PBS algorithm from Ref.\cite{dantsin2002deterministic}, and establish a mapping from $\mathbf{s}$ to $V(\mathbf{s})$. Then, verify whether the new assignment satisfies $F$. To minimize space usage, quantize these two processes, and then integrate them into the amplitude amplification algorithm to form the final quantum algorithm for solving the PBS problem. 
\begin{algorithm}
\caption{Quantum PBS algorithm }
\SetAlgoLined
\label{3QPBS}
\KwIn{ $3$-SAT formula $F$, $r$, assignment $\mathbf{x}_i$.}
\KwOut{ A satisfying assignment of $F$, or \textbf{FALSE}.}
 Prepare initial state 
        $\frac{1}{\sqrt{3^{{r}}}}\sum_{\mathbf{s}\in\{1,2,3\}^{{r}}} |\mathbf{s}\rangle|V(\mathbf{s})\rangle\left|F(\mathbf{x}_{V(\mathbf{s})})\right\rangle$\;
     Apply amplitude amplification to find $V(\mathbf{s})$ 
        such that $F(\mathbf{x}_{V(\mathbf{s})})=1$\;
     Measure and determine $V(\mathbf{s})$ to obtain $\mathbf{x}_{V(\mathbf{s})}$\;   
 Check whether $\mathbf{x}_{V(\mathbf{s})}$ satisfies $F$. If it does, output $\mathbf{x}_{V(\mathbf{s})}$; otherwise, output  \textbf{FALSE}.
\end{algorithm}

Given a center assignment $\mathbf{x} \in \{0,1\}^n$ and a flip sequence 
$\mathbf{s} =\left(s_1, s_2, \cdots, \\s_{r}\right)$.  The flip sequence arises from the procedure used in solving the PBS problem, where at each step a variable is randomly chosen from an unsatisfied clause and flipped, thereby reducing the Hamming ball radius by one. For the 3-SAT problem, each clause contains three literals, corresponding to three possible choices, i.e., $s_i\in\{1,2,3\}$. By making consecutive $s_i$ selections, a mapping from $\mathbf{s}$ to $V(\mathbf{s})$ can be constructed.
Let $V(\mathbf{s}) \subseteq \{1,2,\cdots, n\}$ denote the set of variable indices 
flipped along the path determined by $\mathbf{s}$. 
The notation $x_{V(\mathbf{s})}$ represents the resulting assignment 
obtained from $\mathbf{x}$ after flipping all variables in $V(\mathbf{s})$:
\[
x_{i,V(\mathbf{s})} =
\begin{cases}
\neg x_i, & \text{if } i \in V(\mathbf{s}),\\[4pt]
x_i, & \text{otherwise.}
\end{cases}
\]
In other words, $V(\mathbf{s})$ specifies which variables are flipped, 
while $x_{V(\mathbf{s})}$ denotes the new candidate assignment obtained 
after applying those flips to $\mathbf{x}$.  Next, the new candidate assignment is verified to determine whether it satisfies the formula. For detailed procedures, see Algorithm \ref{3QPBS}.  Algorithm \ref{3QPBS} is referred to as $Q(F, \mathbf{x}_i, {r})$ in the following.

\begin{algorithm}
\caption{QCPBS algorithm }
\SetAlgoLined
\label{CPBS}
\KwIn{ $3$-SAT formula $F$, $r$, assignment $\mathbf{x}_i$, $r_{\max}$.}
\KwOut{ A satisfying assignment of $F$, or \textbf{FALSE}.}
If $\mathbf{x}_i$ satisfies $F$, then output $\mathbf{x}_i$\;
  If $r = 0$, then output \textbf{FALSE}\;
 \If{$r>r_{\max}$} 
 {Select an arbitrary clause 
$C$ of the formula 
$F$ that is unsatisfied by the assignment 
$\mathbf{x}_i$\;
For each literal 
$l$ appearing in the clause 
$C$, construct a new instance by fixing the variable of 
$l$ so that 
$l$ becomes satisfied, resulting in the restricted formula 
$F|_{l=1}$\;
For each such restriction, recursively invoke QCPBS algorithm on $(F|_{l=1},\mathbf{x}_i,r-1, r_{\max})$}
\Else   
{call Algorithm \ref{3QPBS} on $(F, \mathbf{x}_i, {r}_{\max})$}
If any recursively call outputs a $\mathbf{x}^\ast$ such that  $F(\mathbf{x}^\ast) = 1$, the algorithm  outputs the assignment; otherwise, it outputs
 \textbf{FALSE}.
\end{algorithm}
Based on the number of qubits available in existing quantum devices and the qubit resources required by Algorithm \ref{3QPBS}, the scale of PBS problem that the device can solve can be calculated, represented by the radius $r_{\max}$. Let the number of qubits of the available quantum device be $cn$, where $c\in (0,1)$.    According to Theorem 1 in the supplementary material, the number of qubits required to prepare initial state $$\frac{1}{\sqrt{3^{r}}}\sum_{\mathbf{s}\in\{1,2,3\}^{r}} |\mathbf{s}\rangle|V(\mathbf{s})\rangle\left|F(\mathbf{x}_{V(\mathbf{s})})\right\rangle$$ is $O(r\log\!\left(\frac{n}{r}\right) + r + \log n)$.
Suppose
the exact scaling of this number of qubits is $Ar log(n/r) + Br + O(log n)$ for constants $A, B > 0$. Let $\beta(c)$ satisfy  $A \beta(c)\log\!\left(\dfrac{1}{\beta(c)}\right) + B\beta(c)=c,$ where $c\in (0,1)$ is related to the scale $cn$
of the given quantum computer. We assume that the radius called by the quantum device is $r_{\max}\approx \beta(c) n.$ 

\begin{algorithm}
\SetAlgoLined
\caption{ Quantum-classical hybrid  algorithm for PBS problem}
\label{Cball}
\KwIn{ $3$-SAT formula $F$,  assignment $\mathbf{x}_i$,  $r$, code $\mathbf{C}\subset\{1,2, 3\}^{t}$ of radius $\frac{t}{3}$, $r_{\max}$.}
\KwOut{ A satisfying assignment of $F$, or \textbf{FALSE}.}
 If $\mathbf{x}_i$ satisfies $F$, then output $\mathbf{x}_i$\;
  If $r = 0$, then output \textbf{FALSE}\;
  Construct a maximal set 
$G$ of pairwise disjoint unsatisfied $3$-clauses of  $F$ unsatisfied by $\mathbf{x}_i$\;
 \If{$|G| \le t$}
{Consider the variable set 
vbl($G$) appearing in the clauses of 
$G$, and enumerate all assignments 
$\eta\in \{0,1\}^{vbl(G)}$\;
For each 
$\eta$,  call Algorithm \ref{CPBS} on $(F|_\eta, \mathbf{x}_i,  r, r_{\max})$} 
 \Else   
{Choose a subset 
$H
=\{C_1,\dots,C_t\} \subseteq G$
 of $t$ pairwise disjoint unsatisfied clauses\;
 For each codeword 
$w \in \mathbf{C}$, construct a modified assignment 
$\mathbf{x}_i[H,w]$ according to 
$w$\;
For each such modification, recursively call algorithm\ref{Cball} on the instance $(F, \mathbf{x}_i[H,w], r - \tfrac{t}{3}, \mathbf{C})$, which means returning to step 2}
      If any recursive branch finds a $\mathbf{x}^\ast$ such that  $F(\mathbf{x}^\ast) = 1$, the algorithm  outputs this assignment; otherwise, it outputs
 \textbf{FALSE}. 
\end{algorithm}
The main idea  is to recursively call itself using a classical algorithm to handle gradually $r$ when the available number of qubits is limited, until it reaches a range that can be managed by the quantum device. We  now describe several  subalgorithms used as building blocks of the main algorithm. First, based on Ref.\cite{dantsin2002deterministic}, a quantum-classical hybrid algorithm(Algorithm \ref{CPBS} ) for solving PBS problem  is proposed. 
Building on this, and incorporating the scheme from Ref.\cite{moser2011full}, a further improved algorithm(Algorithm \ref{Cball} ) for this problem is presented. 
 In Algorithm \ref{Cball}, $t$ is a function that increases slowly with $r$. In practical instances, $t = \log \log r$. If the last recursion occurs in the case where $G \le t$, the called quantum algorithm solves for a radius of $r_{\max}$; however, if the last recursion occurs in the case where $G > t$, the solution radius is $\bar{r}_{\max}\in(r_{\max}-\frac{t}{3}, r_{\max}]$.
 
Using the preceding subalgorithms, the entire detailed process for solving the 3-SAT problem  is provided in Algorithm \ref{alg:hybridqball}.    The results of Ref.\cite{dunjko2018computational} are represented by the following theorem.

\begin{algorithm}
\SetAlgoLined
\caption{ Quantum-classical hybrid  algorithm for 3-SAT problem}
\label{alg:hybridqball}
\KwIn{ $3$-SAT formula $F$,   covering code $C = \{\mathbf{x}_1, \mathbf{x}_2, \dots, \mathbf{x}_t\}$
$\subset\{0,1\}^{n}$ of radius $r$, code $\mathbf{C}\subset\{1,2, 3\}^{t}$ of radius $\frac{t}{3}$,  $r_{\max}$.}
\KwOut{ A satisfying assignment of $F$, or \textbf{FALSE}.}
 Select an element  $ \mathbf{x}_i\in C$ uniformly at random  and without repetition\;
  \If{$r>r_{\max}$} 
  {Call  Algorithm \ref{Cball} on $(F, \mathbf{x}_i, {r}, r_{\max}, \mathbf{C})$}
  \Else 
 {Call Algorithm \ref{3QPBS} on $(F, \mathbf{x}_i, \bar{r}_{\max})$}
    If any calls finds a $\mathbf{x}^\ast$ such that  $F(\mathbf{x}^\ast) = 1$, the algorithm  outputs this assignment; otherwise, it returns to step 1\; 
When all codewords is exhausted without finding a satisfying assignment, the algorithm outputs \textbf{FALSE}.
\end{algorithm}

\begin{theorem}\label{1}
Given a quantum computer with
$M= cn $, $c\in (0,1)$ is an arbitrary constant, 
  there exists a hybrid quantum-classical algorithm that solves 3-SAT probiem in runtime $ O^*(2^{(0.415 - f(c) + \varepsilon)n})$, where $f(c) \approx0.21\beta(c)>0$ is a constant and $ \varepsilon $ can be made arbitrarily small, $\beta(c) = \Theta(c/\log(1/c))$, using  $O(r\log\!\left(\frac{n}{r}\right) + r + \log n)$ qubits.
\end{theorem}
In the Theorem \ref{1}, $O^*$ denotes ignoring polynomial factors, and $ \Theta$ represents asymptotic equality. 
It is worth emphasizing that this hybrid algorithm does not directly invoke quantum subroutines on balls of arbitrary radius. When the search radius $r $ exceeds the maximum radius $r_{\max}$ that the quantum device can accommodate, the algorithm continues to reduce $r$  until it is reduced to a range feasible for quantum resources. 
\begin{remark}
 Algorithm \ref{CPBS}, \ref{Cball}, and \ref{alg:hybridqball} are based on the design framework proposed in Ref.\cite{dunjko2018computational}, and it has been supplemented  by us.
\end{remark}

\subsection{ Fixed-point quantum search algorithm}
The fixed-point quantum search algorithm is essentially a solution proposed for the unordered database search problem. This algorithm not only avoids the "soufflé problem" that commonly occurs in Grover's algorithm during unordered search but also maintains the quadratic speedup advantage of Grover's algorithm\cite{Grover1996}. Below is the formal definition of the problem that the fixed-point quantum search algorithm aims to solve.
\begin {definition}
Given a Boolean function $F: \{0,1\}^n\rightarrow\{0,1\}$ that contains a number of target items satisfying  $F(x) = 1$. Let $\lambda$  represent the fraction of target states within the total, where $\lambda$ is unknown but it is known that  $\lambda_{min}$  satisfies $\lambda\ge\lambda_{min}>0$. The problem is to find one target item $ x $   with fewer  queries and a success probability of at least $1-\epsilon^2$, where $\epsilon\in(0,1)$.
\end{definition}

Given a unitary operator  $\mathcal{A}$ that prepares the initial state $\ket{S}$, i.e., $\mathcal{A}\ket{0}=\ket{S}$.  Let  \begin{equation}
\ket{S}=\ket{T}+\ket{\bar{T}},
\end{equation} 
with  \begin{equation}
\ket{T} = \sum_{x:F(x)=1}\alpha_x|x\rangle,\quad  \ket{\bar{T}}= \sum_{x:F(x)=0}\alpha_x|x\rangle.
\end{equation} 
It is required to find the target state within $\ket{T}$ with a success probability $1-\epsilon^2$.

The operator for the fixed-point quantum search algorithm is defined as\begin{equation}
G(\alpha,\beta) = - \mathcal{A}U_{0}(\alpha)\mathcal{A}^{\dagger} U_{F}(\beta),
\end{equation}
     \begin{align*}
U_{F}(\beta)|x\rangle &=
\begin{cases}
e^{\imath\beta
\cdot F(x)}|x\rangle, & F(x)=1,\\
|x\rangle, & F(x)=0,
\end{cases} 
 \end{align*}
  \begin{align*}
U_{0}(\alpha)|x\rangle&=
\begin{cases}
e^{\imath\alpha
}|x\rangle, & x=0^n,\\
|x\rangle, & x \neq 0^n.
\end{cases} 
 \end{align*}
 In order to ultimately output the target element with probability $1-\epsilon^2$, there are specific requirements for the number of iterations and the rotation angles $\beta$ and $\alpha$.      Set
$$
L \ge \frac{\log(2/\epsilon)}{\sqrt{{\lambda_{min}}}},
$$
 and require
$
L$  to be the nearest odd integer. Let $$
S_L = G(\alpha_l,\beta_l)\cdots G(\alpha_1,\beta_1)
    = \prod_{j=1}^{l} G(\alpha_j,\beta_j), \quad L=2l+1.
$$
All rotation angles satisfy
 \begin{equation}\label{t3}
\alpha_j = -\beta_{l-j+1}
= 2 \cot^{-1}\!\left( 
    \tan\!\left(\frac{2\pi j}{L}\right)
    \sqrt{1 - \gamma^2}
\right), 
\end{equation}
where $j = 1, 2, \cdots, l$, $\gamma^{-1} = T_{1/L}(1/\epsilon)$. $T$ is the Chebyshev polynomial of the first kind\cite{rivlin2020chebyshev}, which is defined as
$$T_n(x) = \cos\!\left(n \arccos x\right), \quad |x| \leq 1.$$
In fact, the iterative operator $G(\alpha,\beta)$ is applied 
$l$ times in total.

After setting all the relevant parameters,  the complete algorithm process is shown in Algorithm \ref{fixed point}.

\begin{algorithm}
\caption{Fixed-point quantum search algorithm}
\label{fixed point}
\KwIn{  operator $\mathcal{A}$, function $F$, oracle $U_{F}$, parameter $\lambda_{min}$, $\epsilon$.}
\KwOut{ $x \in \{0,1\}^{n}$ such that $F(x)=1$.}
 Set
$L \ge \frac{\log(2/\epsilon)}{\sqrt{{\lambda_{min}}}},$ require
$L$  to be the nearest odd integer\;
 $l=\frac{L-1}{2}$\;
 Applying the $ S_L$ operator to the initial state $\ket{S}$, i.e., $ S_L\ket{S}$, where $$
\ket{S}=\mathcal{A}\ket{0}, \quad S_L 
    = \prod_{j=1}^{l} G(\alpha_j,\beta_j)\;$$
    
 Measure  and obtain an $x$ with $F(x)=1$\;
\end{algorithm}

The following are the results from their paper, which we present in the form of a theorem.
\begin{theorem}\label{2}
Let the target item  $x\in\{0,1\}^n$  satisfy  
$F(x) = 1$, and suppose this target set occupies a fraction $\lambda\ge\lambda_{min}>0$ of the entire search space. Given a parameter $\epsilon\in(0,1)$, Algorithm \ref{fixed point} can find a target item with probability at least $1-\epsilon^2$.
\end{theorem}

\section{Distributed quantum-classical hybrid algorithm for solving $K$-SAT problem}\label{3}
In this section, we first present some parameters related to the algorithm, then design the distributed quantum-classical hybrid algorithm, subsequently prove its correctness, and finally compare it with existing distributed quantum algorithms.
\subsection{ The design of distributed quantum-classical hybrid algorithms}
Given a ${K}$-SAT formula $F=\land_{j =1}^m C_j$, where $C_j=\lor_{i =1}^n l_i$ and $l_i$ is either a Boolean variable $x_i$ or its negation $\neg{x_i}$.  
We computer the frequency of each literal in the entire formula, represented by the following metric. 
\begin {equation}
M({x_i},F)= \sum_{j=1}^{m} \sum_{l_i\in C_j} \mathbf{I}(l_i =x_i \lor l_i = \neg x_i)
\end{equation}
where $\mathbf{I}$ is the indicator function, which takes the value 1 when the condition is satisfied, and 0 otherwise. During circuit implementation, if variable $x_i$ appears frequently in  $F$, more control bits are inevitably required during the computation. To reduce  the depth  of control gates, we decompose $F$ based on the top $k$ variables with the highest frequency of occurrence, which means selecting a decomposition criterion with a larger $M$. If $k$ is small, we directly enumerate all subfunctions. If $k$ is relatively large, we can construct covering codes and construct corresponding subfunctions for each codeword in the covering code, thereby reducing the number of subfunctions. Covering code construction is classical, and its construction cost depends on the radius and the code length. This paper mainly considers the case where $k$ is not large, and directly enumerates all subfunctions.

For example,  
\begin {equation}F = (x_1 \lor x_2 \lor x_4)\land  (x_2 \lor x_3 \lor x_4)\land  (\neg x_1 \lor x_2 \lor \neg x_4). \end{equation}
We have   \begin {equation}
M({x_1},F)=2, M({x_2},F)=3, M({x_3},F)=1, M({x_4},F)=3.
\end{equation}

 For a given ${K}$-SAT formula $F=\land_{j =1}^m C_j$, we  first 
 calculate the $M$ value for each variable and sort them in descending order. The entire sorting can be completed using the quicksort algorithm in $O(n\log n)$. In the computation process, since $F$ has $n$ variables and $m$ clauses, and each variable is represented by a binary number, the time complexity is $O(m\log n)$. As can be seen, this can be accomplished within polynomial time. Next, select the top $k\ge1$ variables for decomposition. If there are multiple identical $M$ values at the last position, we can randomly choose one of the variables. Without loss of generality, assume that \begin {equation}M({x_1},F) \ge M({x_2},F) \ge \cdots \ge M({x_k},F)\cdots \ge M({x_n},F),\end{equation} and let ${i'} = {x_1}{x_2}\cdots{x_k} \in \{0, 1\}^k$. For  $F(x_1x_2\cdots x_n )$, we have \begin {equation}F({i'}x_{k+1}\cdots x_n) = F_{i'}(x_{k+1}\cdots x_n), {i'} = {x_1}{x_2}\cdots{x_k} \in \{0, 1\}^k\end{equation} which is equivalent to decomposing $F$ into $2^k$ subfunctions $F_{i'}$. 
 \begin{algorithm}
\caption{Quantum KPBS algorithm}
\SetAlgoLined
\label{KQPBS}
\KwIn{ $K$-SAT formula $F$($K\ge3$), $r$, assignment $\mathbf{x}_i$, $\epsilon\in(0,1)$.}
\KwOut{ A satisfying assignment of $F$, or \textbf{FALSE}.}
 Prepare initial state 
        $\frac{1}{\sqrt{K^{{r}}}}\sum_{\mathbf{s}\in\{1,2,\cdots,K\}^{{r}}} |\mathbf{s}\rangle|V(\mathbf{s})\rangle\left|F(\mathbf{x}_{V(\mathbf{s})})\right\rangle$\;
         Apply fixed-point search algorithm to find $V(\mathbf{s})$ 
        satisfying $F(\mathbf{x}_{V(\mathbf{s})})=1$ with a success probability of at least $1-\epsilon^2$\;
     Measure and determine $V(\mathbf{s})$ to obtain $\mathbf{x}_{V(\mathbf{s})}$\;    
 Check whether $\mathbf{x}_{V(\mathbf{s})}$ satisfies $F$. If it does, output $\mathbf{x}_{V(\mathbf{s})}$; otherwise, output  \textbf{FALSE}.
\end{algorithm}

For any  $F_{i'}$, we use the covering code to reduce the search for an assignment in $\{0, 1\}^{n-k}$ to a PBS problem.  
We  present a lemma from Ref.\cite{dantsin2002deterministic}, which reveals the existence of covering codes.
\begin{lemma}\label{t4}
Let $d \ge 2$ be a divisor of $n-k \ge 1$, \[
h(\rho) = - \rho \log \rho \;-\; (1-\rho)\log(1-\rho),
\]
and $0 < \rho < 1/2$.
Then there exists a polynomial $q_d$ such that a covering code of length $n-k$,
radius at most $\rho (n-k)$, and size at most
\[
q_d(n-k) \cdot 2^{(1 - h(\rho))(n-k)}
\]
can be constructed in time
\[
q_d(n-k) \cdot \left( 2^{3n/d} + 2^{(1 - h(\rho))(n-k)} \right).
\]
\end{lemma}
Thus, covering codes are constructed according to  Lemma \ref{t4}. Let  covering code $C = \{\mathbf{x}_1, \mathbf{x}_2, \dots, \mathbf{x}_m\}\subset\{0,1\}^{n-k}$ with radius $r$ such that:
$$
\forall x \in \{0,1\}^{n-k}, \exists \mathbf{x}_i \in C, d_H(x, \mathbf{x}_i) \le r.
$$
\begin{algorithm}[H]
\caption{KQCPBS algorithm}
\label{QPBS}
\KwIn{ ${K}$-SAT formula $F$($K\ge3$), $r$, $r_{\max}$, assignment $\mathbf{x}_i$, $\epsilon\in(0,1)$.}
\KwOut{ A satisfying assignment of $F$, or \textbf{FALSE}.}
If $\mathbf{x}_i$ satisfies $F$, then output $\mathbf{x}_i$\;
  If $r = 0$, then output \textbf{FALSE}\;
  \If{$r>r_{\max}$} 
   {Select an arbitrary clause 
$C$ of the formula 
$F$ that is unsatisfied by the assignment 
$\mathbf{x}_i$\;
For each literal 
$l$ appearing in the clause 
$C$, construct a new instance by fixing the variable of 
$l$ so that 
$l$ becomes satisfied, resulting in the restricted formula 
$F|_{l=1}$\;
Sort all formulas $F|_{l=1}$ in ascending order based on the number of unsatisfied clauses\;
Select the one with the fewest remaining unsatisfied clauses, recursively invoke Algorithm \ref{QPBS} on $(F|_{l=1},\mathbf{x}_i,r-1, r_{\max}, \epsilon)$. If the algorithm outputs \textbf{FALSE}, then proceed to try other $F|_{l=1}$ in order.} 
 \Else 
 { Call Algorithm \ref{KQPBS} on ($F, \mathbf{x}_i, r_{\max}, \epsilon$)} 
 If any recursively call outputs a $\mathbf{x}^\ast$ such that  $F(\mathbf{x}^\ast) = 1$, the algorithm  outputs the assignment; otherwise, it outputs
 \textbf{FALSE}.
\end{algorithm}

We  present some subalgorithm for clarity, and assemble them into the main algorithm afterward. First, we present Algorithm \ref{KQPBS}, a quantum algorithm for solving  PBS problem.
In line 1 of the Algorithm \ref{KQPBS}, the initial state  consists of three registers:
the first register encodes the sequential choices 
$\mathbf{s}=s_1, s_2, \cdots, s_{\bar{r}}\in\{1,2,\cdots,K\}$ of flipped literals,
the second register $V(\mathbf{s})$ stores the indices of the variables selected by those choices, and the third register holds the value of $F$ under the newly assigned variables.

 Consider a system consisting of $2^k$ quantum computers, each with a qubit count of  $P = cn  (\text{where } c\in (0,1)).$ To solve $K$-SAT problem on this system, the first step is to evaluate the resource requirement: namely, the  number of qubits needed to solve a problem instance of radius $r$ using the Algorithm \ref{KQPBS}. 
  Based on the encoding and computing scheme in  Ref. \cite{dunjko2018computational} and in comparison with the qubit capacity of the available computers, we can determine the maximum problem radius that this system can handle, denoted as $r_{\max}$. For the subsequent development of a distributed hybrid algorithm, we preset the solution radius here as $r_{\max}$.

  A heuristic strategy is proposed in this paper to prioritize the selection of assignment schemes for which satisfying conditions can be more readily found.  
 Throughout the subsequent subalgorithms, this strategy will remain consistently applied.  First, we introduce subalgorithm \ref{QPBS}, a quantum-classical hybrid algorithm designed for solving PBS problem. Subsequently, we present a more efficient quantum-classical hybrid algorithm, namely Algorithm \ref{fball}, which invokes subalgorithm \ref{QPBS} as a subroutine during its execution.

Additionally, for the $K$-ary covering code $\mathbf{C}$ in Algorithm \ref{fball}, we adopt the definition from Ref.\cite{moser2011full}.
\begin{definition}
Let $s, t \in \mathbb{N}$, $$B_s^{(K)}(w) := \{\, w' \in \{1,\ldots,K\}^t \mid d_H(w, w') \le s \,\}.$$ If
$$
  \bigcup_{w \in \mathbf{C}} B_s^{(K)}(w) = \{1,\ldots,K\}^t,
$$a set $\mathbf{C} \subseteq \{1,\ldots,K\}^t$ is called a code of
covering radius $s$. 
In other words, for each $w' \in \{1,\ldots,K\}^t$, there is some
$w \in \mathbf{C}$ such that $d_H(w,w') \le s$.
\end{definition}
The following lemma from Ref.\cite{moser2011full} shows the existence of $K$-ary covering codes.
\begin{lemma}\label{t5}
For any $t, K \in \mathbb{N}$ and $0 \le s \le t$, there exists a
code $\mathbf{C} \subseteq \{1,\ldots,K\}^t$ of covering radius $s$ such that
$$
|\mathbf{C}| \le 
\left\lceil 
\frac{t \ln(K)\, K^t}{\binom{t}{s} (K-1)^s}
\right\rceil.
$$
\end{lemma}

\begin{algorithm}[H]
\caption{ Quantum-classical hybrid algorithm for ${K}$-PBS problem}
\label{fball}
\KwIn{ ${K}$-SAT formula $F$($K\ge3$),  assignment $\mathbf{x}_i$, $r$,
 code $\mathbf{C}\subset\{1,2,\cdots,K\}^{t}$ of radius $\frac{t}{K}$,  $r_{\max}$, $\epsilon\in(0,1)$.}
\KwOut{ A satisfying assignment of $F$, or \textbf{FALSE}.}
If $\mathbf{x}_i$ satisfies $F$, then output $\mathbf{x}_i$\;
  If $r = 0$, then output \textbf{FALSE}\;
 Construct a maximal set 
$G$ of pairwise disjoint unsatisfied $K$-clauses of  $F$ unsatisfied by $\mathbf{x}_i$\;
 \If{$|G| \le t$}
{
Consider the variable set 
vbl($G$) appearing in the clauses of 
$G$, and enumerate all assignments 
$\phi\in \{0,1\}^{vbl(G)}$\;
Sort all formulas $F|_\phi$ in ascending order based on the number of unsatisfied clauses\;
Select the one with the fewest remaining unsatisfied clauses, and then invoke Algorithm \ref{QPBS} on  $(F|_\phi, \mathbf{x}_i, r, r_{\max}, \epsilon)$. If the algorithm outputs \textbf{FALSE}, then proceed to try other $F|_\phi$ in order.
}
\Else
{Choose a subset 
$H
=\{C_1,\dots,C_t\} \subseteq G$
 of $t$ pairwise disjoint unsatisfied clauses\;
 For each codeword 
$w \in \mathbf{C}$, construct a modified assignment 
$\mathbf{x}_i[H,w]$ according to 
$w$\;
For each such modification, 
record the number of clauses that $\mathbf{x}_i[H,w]$ fails to satisfy under $F_{i'}$\;
Select the one with the fewest remaining unsatisfied clauses, then recursively call Algorithm \ref{fball} on  $(F, \mathbf{x}_i[H,w], r - \frac{t}{K},  r_{\max},\mathbf{C}, \epsilon)$. If the algorithm outputs \textbf{FALSE}, then proceed to try other $\mathbf{x}_i[H,w]$ in order.}
 If any recursive branch finds a $\mathbf{x}^\ast$ such that  $F(\mathbf{x}^\ast) = 1$, the algorithm  outputs this assignment; otherwise, it outputs
 \textbf{FALSE}.      
\end{algorithm}
The set $G$ mentioned in line 3 of Algorithm \ref{fball}, defined as the maximum pairwise disjoint set of clauses, consists of clauses that share no variables with one another. Every remaining unsatisfied clause in $F$ must share at least one variable with some clause in $G$. We provide simplification rules for partial assignments. Substitute $\phi$ into the Boolean formula $F$; if there exists a clause in which all Boolean variables of its literals are 0, discard this branch $F|_\phi$ immediately.   All clauses related to the Boolean variable $\phi$ will participate in the simplification. Satisfied clauses are entirely removed, while unsatisfied clauses are retained, but all parts related to this variable $\phi$—whether positive or negative—are removed. Whenever an partial assignment  occurs, we always prioritize the one with the smallest unsatisfied clause.

\begin{algorithm}[H]
\caption{Distributed quantum-classical hybrid algorithm for ${K}$-SAT problem}
\label{fastball}
\KwIn{ ${K}$-SAT formula $F$($K\ge3$),  covering code $C = \{\mathbf{x}_1, \mathbf{x}_2, \dots, \mathbf{x}_m\}\subset\{0,1\}^{n-k}$ of radius $r$, 
 code $\mathbf{C}\subset\{1,2,\cdots,K\}^{t}$ of radius $\frac{t}{K}$,  $r_{\max}$, $\epsilon\in(0,1)$.}
\KwOut{ A satisfying assignment of $F$, or \textbf{FALSE}.}
  Decompose $F$ into $2^k$ subfunctions $F_{i'}$ based on the $M$ values\; Randomly and non-repetitively select a $F_{i'}$ to proceed to the next step\;
  Check whether all codewords in $C$ satisfy $F_{i'}$. Once an $\mathbf{x}_i$ satisfying $F_{i'}$ appears, the algorithm outputs $\mathbf{x}_i$; otherwise, the number of corresponding unsatisfied clauses is recorded\;
Sort the recorded codewords in ascending order based on the number of unsatisfied clauses\;
 Select the top $2^k$ distinct codewords for parallel processing\;
  \If{$r>r_{\max}$} 
  {Call  Algorithm \ref{fball} on $(F_{i'}, \mathbf{x}_i, {r}, r_{\max}, \mathbf{C},   \epsilon)$}
  \Else 
 {Call Algorithm \ref{KQPBS} on $(F_{i'}, \mathbf{x}_i, \bar{r}, \epsilon)$}
    If any calls finds a $\mathbf{x}^\ast$ such that  $F_{i'}(\mathbf{x}^\ast) = 1$, the algorithm  outputs this assignment; otherwise, it returns to step 5\; 
 If all codewords in step 5 have been attempted, the algorithm will backtrack to step 2 to reselect the $F_{i'}$\; 
The algorithm  outputs \textbf{FALSE} when all $F_{i'}$ is exhausted without finding a satisfying assignment.
\end{algorithm}
Let $s=\frac{t}{K}$.  A fixed ordering is first established for all literals in each clause of $H$. Subsequently, $\mathbf{x}_i[H,w]$ represents the flipping of the $w_j$-th literal in clause $C_j$ of $H$ with respect to assignment $\mathbf{x}_i$, for $1\le j\le t$. Next, an example is given. Assuming $t=3$ and $$H=\{ (\neg x_1 \lor \neg x_2 \lor \neg x_3),  (\neg x_4 \lor \neg x_5 \lor \neg x_6),  ( \neg x_7 \lor \neg x_8 \lor \neg x_9)\}.$$ Let $w=(1,2,3)$; $\mathbf{x}_i$ denotes the all-one assignment. Then, $\mathbf{x}_i[H,w]$ denotes setting the variables $x_1, x_5, x_9$ to $0$, while keeping the assignments of all other variables at $1$. After $H$ is understood, replace $\mathbf{x}_i[H,w]$ with $\mathbf{x}_i[w]$. 

Next, we examine the change in the radius $r$ of the PBS problem after one call to the  code $\mathbf{C}$. Suppose $x^*\in \{0,1\}^{n-k}$ is the promised satisfying assignment; then we must have $d_H(x^*, x_i) \le r$ for any $\forall x_i \in B_r(\mathbf{x}_i)$. Based on the property of the covering code, there  exists a $w^*\in \{1,\ldots,K\}^t $ such that $x_i[w^*]$ moves closer to a satisfying assignment and satisfies $d_H(x_i[w^*], x^*) =d_H(x^*, x_i)-t\le r-t$. Here, $w^*$ does not necessarily belong to the code $\mathbf{C}$, but it must be covered by $\mathbf{C}$. For any $w \in \{1,\ldots,K\}^t$, we have $d_H(w^*, w)\le\frac{t}{K}$ and $d_H(x_i[w^*], x_i[w])=2d_H(w^*, w)$. Due to the triangle inequality property of the Hamming distance, we have
\begin{align*}d_H(x_i[w], x^*)\le &d_H(x_i[w^*], x_i[w])+d_H(x_i[w^*], x^*)\\
\le &r-t+\frac{2t}{K}.\end{align*}
Hence, set $ \Delta := t - \tfrac{2t}{K}$. Let $t$ be a function that grows slowly with $r$. In practical instances, one may set $t = \log \log r$. 

\begin{figure*}[h]
		\centering
		\includegraphics[width=1\linewidth]{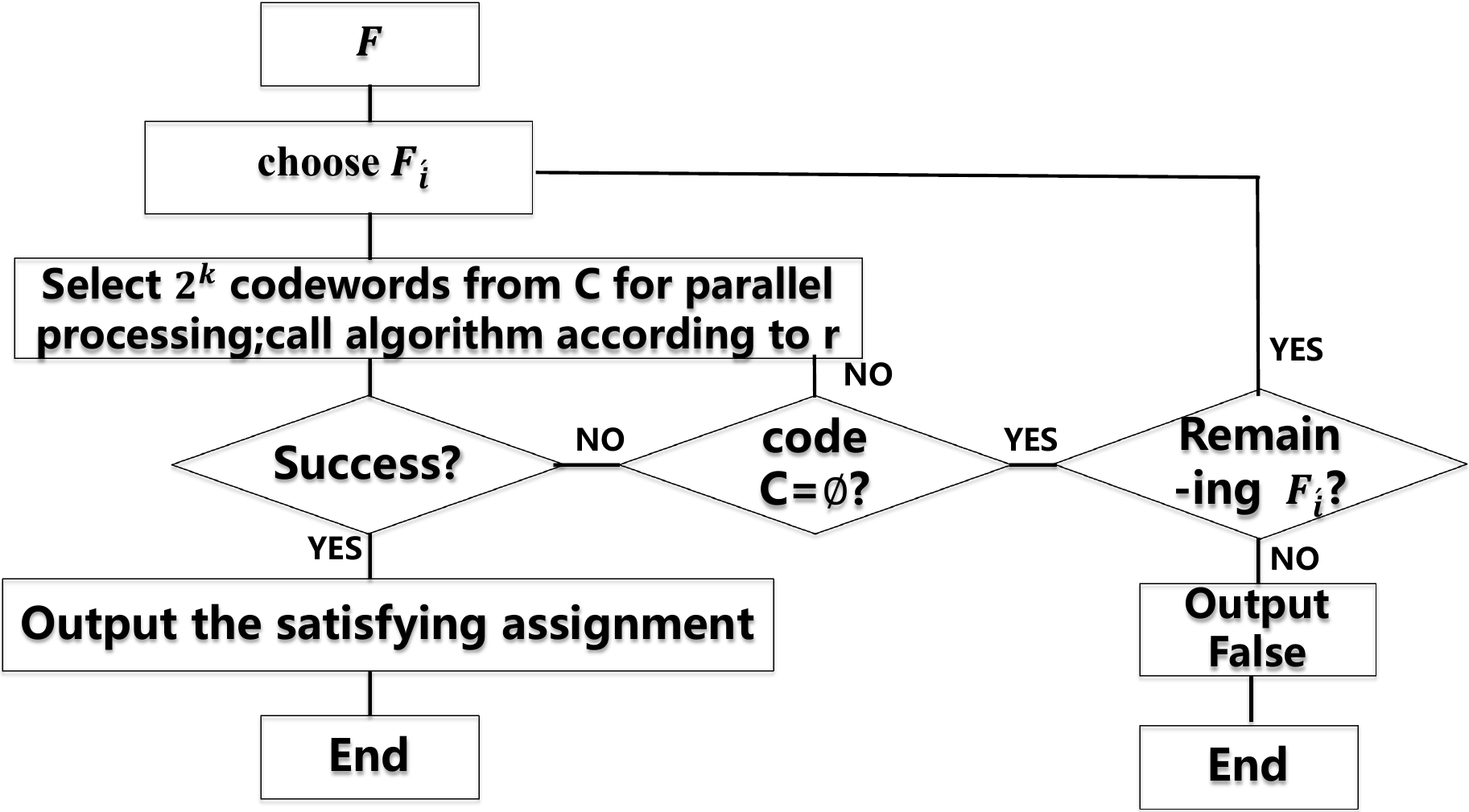}
		\setlength{\abovecaptionskip}{-0.01cm}
		\caption{Framework diagram  of  Algorithm \ref{fastball}.}
		\label{A0}
	\end{figure*}

On this basis, we formally present the distributed quantum-classical hybrid algorithm \ref{fastball} for solving ${K}$-SAT problem. We random select the instance $F_{i'}$. Then, based on the constructed covering code, substitute each codeword into this $F_{i'}$, select the $2^k$ codewords with the smallest number of unsatisfied clauses, and search the Hamming balls corresponding to these $2^k$ codewords to determine whether a satisfying assignment can be obtained. If not, reselect another set of $2^k$ codewords for searching. The process continues until the  set $C$ becomes empty.  The overall  framework of the distributed quantum-classical hybrid algorithm is briefly illustrated in Fig. \ref{A0}.

Algorithm \ref{fastball} can essentially be divided into two processes: an outer loop and an inner loop. The outer loop iterates over different $F_{i'}$, while the inner loop operates over covering code $C$. During Algorithm \ref{fastball}'s invocation of Algorithm \ref{fball}, when the radius $r$ decreases to a point where the quantum device can be invoked, the call radius at that point is denoted as $\bar{r}$, satisfying $r_{\max}\ge\bar{r} \ge r_{\max}- \Delta$.

\subsection{ Correctness analysis of the algorithm}
The initial state in Algorithm \ref{KQPBS} can be prepared, and its concrete implementation can be found in Section A of the supplementary material of Ref. \cite{dunjko2018computational}.
The only difference is that our clauses contain up to 
$K$ literals, so selecting the variable to flip incurs an additional 
$O(K)$ overhead.
However, since 
$K$ is a fixed constant, this overhead is negligible in the asymptotic sense, and therefore we have the following proposition.
\begin{proposition}\label{t1}\rm{[State preparation for the PBS problem oracle]}
There exists a quantum circuit that prepares the  state
$$
    \frac{1}{\sqrt{K^r}}\sum_{\mathbf{s}\in\{1,2,\cdots,K\}^r}
        |\mathbf{s}\rangle\,|V(\mathbf{s})\rangle\,|F_{i'}(\mathbf{x}_{V(\mathbf{s})})\rangle,
$$
using at most
$$
    O\!\left(r\log\!\frac{n-k}{r}+r+\log (n-k)\right)
$$
qubits  and $O(\mathrm{poly}(n-k))$ gates.
\end{proposition}
\begin{proof}
Consider that the initialized quantum state is
$$
    \frac{1}{\sqrt{K^r}}\sum_{\mathbf{s}\in\{1,2,\cdots,K\}^r} 
        |\mathbf{s}\rangle\,|0\rangle\,|0\rangle.
$$
It can be seen that the first register requires at most $r\lceil \log K \rceil$ qubits. $K$ being a positive integer, it follows that $r\lceil \log K \rceil\in O\!\left(r\right)$.
According to Proposition 2 in the supplementary material of Ref. \cite{dunjko2018computational}, there exists a unitary operator 
$U_1$ acting on the first two registers such that 
\begin{align*}
   &U_1 \left(\frac{1}{\sqrt{K^r}}\sum_{\mathbf{s}\in\{1,2,\cdots,K\}^r} 
        |\mathbf{s}\rangle\,|0\rangle\,|0\rangle\right)\\=&\frac{1}{\sqrt{K^r}}\sum_{\mathbf{s}\in\{1,2,\cdots,K\}^r} 
        |\mathbf{s}\rangle\,|V(\mathbf{s})\rangle\,|0\rangle.
\end{align*}
The implementation of this operator requires 
$O\!\left(r\log\!\frac{n-k}{r}+r+\log (n-k)\right)$ qubits and can be realized using 
$O(\mathrm{poly}(n-k))$  quantum gates.
Furthermore, Proposition 3 in the same supplementary material guarantees the existence of a unitary operator 
$U_2$ acting on the last two registers such that
\begin{align*}
  & U_2\left(\frac{1}{\sqrt{K^r}}\sum_{\mathbf{s}\in\{1,2,\cdots,K\}^r} 
        |\mathbf{s}\rangle\,|V(\mathbf{s})\rangle\,|0\rangle\right)\\
        =&\frac{1}{\sqrt{K^r}}\sum_{\mathbf{s}\in\{1,2,\cdots,K\}^r} 
        |\mathbf{s}\rangle\,|V(\mathbf{s})\rangle\,|F_{i'}(\mathbf{x}_{V(\mathbf{s})})\rangle.
\end{align*}
Its implementation likewise uses 
$O\!\left(\log (n-k)\right)$ qubits and can be constructed with 
$O(\mathrm{poly}(n-k))$  quantum gates. By combining the numbers of qubits and quantum gates required by $U_1$ and $U_2$, we obtain the result.
\end{proof}

Given parameters $ n, k, r$, according to Proposition \ref{t1}, suppose the number of qubits is $N = Ar\log\!\frac{n-k}{r}+Br+O(\log (n-k))$, where $A, B > 0$. Once the encoding and computing scheme from $|\mathbf{s}\rangle$ to $|V(\mathbf{s})\rangle$ and then to $|F_{i'}(\mathbf{x}_{V(\mathbf{s})})\rangle$ is determined, the corresponding value of $N$ is also determined.
The current number of qubits in existing quantum computers is $cn$. Let $Ar\log\!\frac{n-k}{r}+Br+O(\log (n-k))=cn$, and solve for $r_{max}$. We may assume that $r_{max}=\gamma(n-k)$. Then we can obtain $A\gamma(n-k)\log\!\frac{1}{\gamma}+B\gamma(n-k)+O(\log (n-k))=cn$. Thus, it follows that $A\gamma\log\!\frac{1}{\gamma}+B\gamma+O(\frac{\log (n-k)}{n-k})=c\frac{n}{n-k}$. When $n$ is sufficiently large, the following holds: $A\gamma\log\!\frac{1}{\gamma}+B\gamma=c$. Therefore, as long as $\gamma$ satisfies $A\gamma\log\!\frac{1}{\gamma}+B\gamma=c$, then $r_{max}\approx\gamma(n-k)$.

Let \begin{equation}\label{e1}
\mathcal{A} \ket{0}=\frac{1}{\sqrt{K^r}}\sum_{\mathbf{s}\in\{1,2,\cdots,K\}^r}
        |\mathbf{s}\rangle\,|V(\mathbf{s})\rangle\,|F_{i'}(\mathbf{x}_{V(\mathbf{s})})\rangle.
\end{equation}
We define the notation
\begin{align*}&\ket{T}:=\frac{1}{\sqrt{K^r}}\sum_{\substack{\mathbf{s}\in\{1,2,\cdots,K\}^r \\ F_{i'}(\mathbf{x}_{V(\mathbf{s})})=1}}
        |\mathbf{s}\rangle\,|V(\mathbf{s})\rangle\,|1\rangle, \\
        &\ket{\bar{T}}:=\frac{1}{\sqrt{K^r}}\sum_{\substack{\mathbf{s}\in\{1,2,\cdots,K\}^r \\ F_{i'}(\mathbf{x}_{V(\mathbf{s})})=0}}
        |\mathbf{s}\rangle\,|V(\mathbf{s})\rangle\,|0\rangle.\end{align*}
\begin{proposition}\label{t2}
Given  oracle $\mathcal{A}$, $\epsilon\in(0,1)$,  
quantum search algorithm can solve PBS problem with a success probability of at least $1-\epsilon^2$, and its time complexity is $O^*(\sqrt{K^r})$.
\end{proposition}
\begin{proof}
 
 Eq.(\ref{e1}) can be reformulated as
\begin{equation}\label{e2}
\frac{1}{\sqrt{K^r}}\sum_{\mathbf{s}\in\{1,2,\cdots,K\}^r}
        |\mathbf{s}\rangle\,|V(\mathbf{s})\rangle\,|F_{i'}(\mathbf{x}_{V(\mathbf{s})})\rangle
        =\ket{T}+\ket{\bar{T}}.
\end{equation}
Given $\epsilon\in(0,1)$. Set $\lambda_{min}=\frac{1}{{K^r}}$, $L\ge {\log(2/\epsilon)}{\sqrt{{K^r}}},$ where $L$ takes the nearest odd number to ${\log(2/\epsilon)}{\sqrt{{K^r}}}$.
The $G$ operator is defined as follows:
 \begin{equation}
G(\alpha,\beta) = - \mathcal{A}U_{0}(\alpha)\mathcal{A}^{\dagger} U_{1}(\beta),
\end{equation}
where  \begin{align*}&U_{1} = \mathcal{I}_{N}-(1-e^{\imath\beta
})\left(\mathcal{I}_{N-1} \otimes\ket{1}\bra{1} \right),\\
&U_0 = \mathbb{I}_{N} - (1-e^{\imath\alpha
}) \ket{0}_{N}\bra{0}_{N}, \end{align*} where $\alpha$ and $\beta$ can be calculated from Eq.(\ref{t3}).  Under this setup, the fixed-point quantum search algorithm finds the target element with probability 
at least $1-\epsilon^2$. The query complexity of the entire quantum algorithm is  $L-1\in O(\sqrt{K^r})$, meaning that the  algorithm requires $L-1$ calls to the  operator $\mathcal{A}$. Combining with Proposition \ref{t1}, we know that the preparation of the initial state, i.e., Eq.(\ref{e1}), only requires polynomial time. Therefore, the overall algorithm runs in time $O^*(\sqrt{K^r})$.
\end{proof}

As can be seen from  Algorithm  \ref{fball} and \ref{fastball}, the quantum computational radius corresponding to the algorithm that can enter the solving phase is different in the two cases of $G \le t$ and $G > t$.
\begin{proposition}\label{t3}
 If $G \le t$, 
then $F_{i'}|_\phi$ is  a Boolean formula with at most $K-1$ literals.
\end{proposition}
\begin{proof}
We mainly explains why the maximum number of clause variables is $K-1$ rather than $K$. 
For any clause $C$ that is unsatisfied under assignment $\mathbf{x}_i$,  
either $C\in G$,  
or $C$ shares at least one variable with some $C_j\in G$.  
Otherwise, if there exists such a $C$ that is unsatisfied under $\mathbf{x}_i$ and is disjoint from all $C_j\in G$, we could simply add it to $G$ to form a larger disjoint set $G\cup\{D\}$,  
which would contradict the assumption that $G$ is a maximal disjoint set.
Therefore, for any clause $C\notin G$ that is unsatisfied under assignment $\mathbf{x}_i$, due to the influence of the values in $\phi$, its number of variables is reduced by at least one. That is, in the worst case, it becomes an unsatisfied clause containing $K-1$ variables. Furthermore, $F_{i'}|_\phi$ is  a Boolean formula with at most $K-1$ literals.
\end{proof}

\begin{theorem}\label{2}
Consider $2^k$ quantum computers with $P=cn$ qubits, where $c\in (0,1)$.  Given  oracle $\mathcal{A}$, $\epsilon\in(0,1)$,  and the binary covering codes and $K$-ary covering codes satisfy  Lemma \ref{t4} and Lemma \ref{t5}, respectively, then Algorithm \ref{fastball} can solve ${K}$-SAT problem with a success probability of at least $1-\epsilon^2$, using at most
$$
    O\!\left(r\log\!\frac{n-k}{r}+r+\log (n-k)\right)
$$
qubits and with a runtime of $O^*\left(2^{\,(n-k)\!\left(
1 +\log\left(\frac{K-1}{ K}\right)-\gamma\log\left(\frac{K-1}{\sqrt K}\right)+\epsilon^{\prime}
   \right)}\right),$ where $\epsilon^{\prime}$ can be made arbitrarily small, $\gamma$ is a constant related to $c$.
\end{theorem}
\begin{proof}
The number of qubits follows directly from Proposition \ref{t1}. Consequently, we shall concentrate our analysis on the running time of the algorithm. Let's first examine the time required to solve a single PBS problem. If entering into $G \le t$, for at most $t$  $K$-clauses in $G$, it can generate $2^{tK}$ branches.  Then, according to the strategy, branches are selected to continue recursion until the call radius $r_{max}$ is satisfied, and the quantum fixed point search algorithm is employed to accelerate the solving process. Combining Proposition \ref{t3}, the total time can be expressed as $O^*(2^{tK}(K-1)^{r-r_{max}}({K-1})^{\frac{r_{max}}{2}})$. Since $K$ and $t$ are constants, when $G \le t$, the total solving time of Algorithm \ref{fball} on the single PBS problem is $O^*\left((K-1)^{r-\frac{r_{max}}{2}}\right)$.

Next, let's examine the second case where $G > t$. According to Lemma \ref{t5}, the size of a $K$-ary code satisfies $|\mathbf{C}| \le 
\left\lceil 
\frac{t \ln(K)\, K^t}{\binom{t}{r} (K-1)^r}
\right\rceil$.  Moreover, Proposition 10 in Ref.\cite{moser2011full} has proven that $|\mathbf{C}|\le t^{2} (K - 1)^{t - 2t/K}$. Since $ \Delta = t - \tfrac{2t}{K}$, the total number of branches is $|\mathbf{C}|\le t^{2} (K - 1)^{\Delta}=(t^{\frac{2}{\Delta}} (K - 1))^{\Delta}$. Since $\lim_{t \to \infty}t^{\frac{2}{\Delta}}=1$, for  $\varepsilon> 0$ and $t$ sufficiently large, we have $t^{\frac{2}{\Delta}} \le 1 + \frac{\epsilon}{K - 1}$, and thus $|\mathbf{C}|\le(K - 1+\epsilon)^{\Delta}$. Without loss of generality, let the radius for invoking the quantum fixed-point search in this case be $\bar{r}$, which satisfies $r_{\max}\ge\bar{r} \ge r_{\max}- \Delta$. Thus, the total time can be expressed as $O^*((K - 1+\epsilon)^{r-\bar{r}}K^{\frac{\bar{r}}{2}})$.

The first two scenarios are both single-state cases: either $G \le t$ holds consistently until the radius decreases to the threshold where quantum computation becomes applicable, or $G > t$ holds consistently throughout. Additionally, there exists a third scenario: initially, the state enters $G > t$, later transitions into $G \le t$), and then the radius decreases to a range that can be processed by a quantum computer. Let the radius at which the transition from $G > t$ to $G \le t$ occurs be $ r^{\prime}$. Then the  time can be expressed as  \begin{align*}
&(K - 1+\epsilon)^{r-r^{\prime}}2^{tK}(K-1)^{r^{\prime}-r_{max}}({K-1})^{\frac{r_{max}}{2}}\\
<& (K - 1+\epsilon)^{r-r_{max}}2^{tK}K^{\frac{r_{max}}{2}}. \end{align*}Thus, the total time is $O^*((K - 1+\epsilon)^{r-\bar{r}}K^{\frac{\bar{r}}{2}})$.

Combining Lemma \ref{t4} and letting $r=\rho (n-k)$, the total time for solving the ${K}$-SAT problem is \begin{equation}\label{t9}
q_d(n-k) \cdot 2^{(1 - h(\rho))(n-k)}O^*((K - 1+\epsilon)^{\rho (n-k)-r_{max}}K^{\frac{r_{max}}{2}}).\end{equation}This time does not account for the construction of the covering code,  because the covering code is first constructed offline before executing the algorithm.

Substituting $r_{max}=\gamma(n-k)$ into Eq.(\ref{t9}) yields
\begin{equation}\label{t10}
O^*\left(2^{\,(n-k)\!\left(
1 -h(\rho)
   \right)}(K - 1+\epsilon)^{\rho({n-k})}
   {\left(\frac{\sqrt K}{K - 1+\epsilon}\right)}^{\gamma(n-k)}\right).
\end{equation}
Since $\frac{\sqrt K}{K - 1+\epsilon}<\frac{\sqrt K}{K - 1}$, it follows that 
\begin{equation}\label{t11}
\begin{aligned}
 {\left(\frac{\sqrt K}{K - 1+\epsilon}\right)}^{\gamma(n-k)}
 < &{\left(\frac{\sqrt K}{K - 1}\right)}^{\gamma(n-k)}\\
 =&2^{\gamma(n-k)\log\left(\frac{\sqrt K}{K - 1}\right)}.
  \end{aligned}
\end{equation}
Additionally, we have
\begin{equation}\label{t12}
 (K - 1+\epsilon)^{\rho({n-k})}=2^{\rho({n-k})\log(K - 1+\epsilon)}.
\end{equation}
Therefore, Eq.(\ref{t10})  can be expressed as \begin{align}\label{t13}&O^*\left(2^{\,(n-k)\!\left(
1 -h(\rho)
   \right)}(K - 1+\epsilon)^{\rho({n-k})}
   {\left(\frac{\sqrt K}{K - 1}\right)}^{\gamma(n-k)}\right)\nonumber\\
 =  &O^*\left(2^{\,(n-k)\!\left(
1 -h(\rho)-\gamma\log\left(\frac{K-1}{\sqrt K}\right)+\rho\log(K - 1+\epsilon)
   \right)}\right). 
   \end{align}
   To minimize $1 -h(\rho)-\gamma\log\left(\frac{K-1}{\sqrt K}\right)+\rho\log(K - 1+\epsilon)$ for $0 < \rho < 1/2$,  let $f(\rho)=1 -h(\rho)+C\rho$, $C=\log(K - 1+\epsilon)$ is a constant greater than 1. By taking the derivative, we can find that the minimum of $f(\rho)$ is attained at $\rho=\frac{1}{2^C+1}$.
  Therefore, setting $\rho= \frac{1}{K}$ minimizes $1 -h(\rho)-\gamma\log\left(\frac{K-1}{\sqrt K}\right)+\rho\log(K - 1+\epsilon)$. 
 We have \begin{align}&1 -h(\frac{1}{K})+\frac{\log(K - 1+\epsilon)}{K}-\gamma\log\left(\frac{K-1}{\sqrt K}\right)\\=&1+\log\left(\frac{K-1}{ K}\right)+\epsilon^{\prime}-\gamma\log\left(\frac{K-1}{\sqrt K}\right), \end{align}
 where $\epsilon^{\prime}=\frac{1}{K}\log(\frac{K - 1+\epsilon}{K-1})$. 
 Ultimately, the time complexity obtained is  $$O^*\left(2^{\,(n-k)\!\left(
1 +\log\left(\frac{K-1}{ K}\right)-\gamma\log\left(\frac{K-1}{\sqrt K}\right)+\epsilon^{\prime}
   \right)}\right).$$
   
  Since each synchronization involves $2^k$ computers, the total time can be expressed as $$O^*\left(2^{\,(n-k)\!\left(
1 +\log\left(\frac{K-1}{ K}\right)-\gamma\log\left(\frac{K-1}{\sqrt K}\right)+\epsilon^{\prime}
   \right)-k}\right).$$Combining the initial $2^k$ $F_{i'}$, the running time for solving ${K}$-SAT problem is $$O^*\left(2^{\,(n-k)\!\left(
1 +\log\left(\frac{K-1}{ K}\right)-\gamma\log\left(\frac{K-1}{\sqrt K}\right)+\epsilon^{\prime}
   \right)}\right).$$
\end{proof}
\subsection{ Comparison to related works}
Compared to the algorithm in Ref.\cite{moser2011full}, this algorithm contributes a constant-factor speedup of $\gamma\log\left(\frac{K-1}{\sqrt K}\right)$ in the exponent—that is, it reduces the complexity from $O(2^n)$ to $O(2^{(1-c)(n-k)})$. When the parameters $K = 3$, $k=0$, the complexity of our algorithm reduces to the case described in Ref.\cite{dunjko2018computational}. When  the parameters $K = 3$, $k\ge1$, compared to the algorithm in Ref.\cite{dunjko2018computational}, our algorithm reduces the number of qubits required by at least $k$. Compared with the distributed quantum algorithm proposed in Ref.\cite{lin2024parallel}, our algorithm completely avoids quantum communication overhead during the distributed computation process, while the compared algorithm relies heavily on quantum communication to achieve  computation.

\section{Conclusion}\label{4}
  Considering the constraints on circuit depth and qubit count, we have proposed a distributed quantum-classical hybrid algorithm for solving the ${K}$-SAT problem, which leverages multiple small-scale quantum computers to accelerate the solving process of classical algorithms. Each time, we select the subproblem that is most likely to contain a solution to search for the target, while making full use of the problem structure by employing quantum computers to tackle classically challenging subprocesses. This approach reduces the dominant exponential term in the runtime by a constant factor that scales linearly with the problem size. The advantage of our proposed algorithm lies in the fact that it not only eliminates the need for quantum communication but also requires only a small-scale quantum computer to run.
  
The Constraint Satisfaction Problem (CSP) is a generalization of the SAT problem, with its core feature being the allowance of multi-valued assignments. Given this generality, the algorithm proposed in this paper can also be applied to constraint satisfaction problems and other, more general problem domains.  In the future, it would be possible to consider implementing the entire algorithm on NISQ quantum devices by constructing well-behaved covering codes, using more efficient encoding schemes, and further reducing the number of qubits required.

\section*{Declaration of competing interest}
The authors declare that they have no known competing financial interests or personal relationships that can have appeared to influence the work reported in this paper.


\bibliographystyle{elsarticle-num} 
 \bibliography{chapter2}

\end{document}